\documentclass[conference]{IEEEtran}

\makeatletter

\newcommand{\Rmnum}[1]{\expandafter\@slowromancap\romannumeral #1@}
\makeatother

\usepackage{graphicx,subfigure,epsfig}
\usepackage{cite,verbatim,multirow}
\usepackage{amssymb,amsmath,amsfonts}
\usepackage{amsthm}
\usepackage{mathrsfs}
\usepackage{float}
\usepackage{soul}
\usepackage{colortbl, color}
\usepackage{cases}
\usepackage[ruled,vlined,linesnumbered]{algorithm2e}

\usepackage{algpseudocode}

\usepackage{subfigure}
\usepackage{url, xcolor}
\usepackage{verbatim, booktabs}

\newtheorem{theorem}{Theorem}
\newtheorem{definition}{Definition}

\newtheorem{lemma}{Lemma}

\newtheorem{property}{Property}

\usepackage{ulem}
\newcommand{\ls}[1]  
   {\dimen0=\fontdimen6\the=#1\dimen0
    \advance\lineskip.5\fontdimen5\the\lineskip-\dimen0
    \lineskiplimit=.9\lineskip
    \baselineskip=\lineskip
    \advance\baselineskip\dimen0
    \normallineskip\lineskip
    \normallineskiplimit\lineskiplimit
    \normalbaselineskip\baselineskip
    \ignorespaces
   }

\DeclareMathOperator*{\argmin}{arg\,min}
\DeclareMathOperator*{\argmax}{arg\,max}

\begin{document}
\bibliographystyle{ieeetr}

\title{BOOST: Base Station ON-OFF Switching Strategy for Energy Efficient Massive MIMO HetNets} 

\author{\authorblockN{Mingjie Feng, Shiwen Mao}
\authorblockA{Dept. Electrical \& Computer Engineering \\ Auburn University, Auburn, AL 36849-5201}
\and
\authorblockN{Tao Jiang}
\authorblockA{School of Electronics \& Information Engineering \\ Huazhong Univ. Science \& Technology, Wuhan 430074 China}
}
\maketitle

\begin{abstract}
In this paper, we investigate the problem of optimal base station (BS) ON-OFF switching and user association in a heterogeneous network (HetNet) with massive MIMO, with the objective to maximize the system energy efficiency (EE). The joint BS ON-OFF switching and user association problem is formulated as an integer programming problem. We first develop a centralized scheme, in which we relax the integer constraints and employ a series of Lagrangian dual methods that transform the original problem into a standard linear programming (LP) problem. Due to the special structure of the LP, we prove that the optimal solution to the relaxed LP is also feasible and optimal to the original problem. We then propose a distributed scheme by formulating a repeated bidding game for users and BS's, and prove that the game converges to a Nash Equilibrium (NE). Simulation studies demonstrate that the proposed schemes can achieve considerable gains in EE over several benchmark schemes in all the scenarios considered.
\end{abstract}


\pagestyle{plain}\thispagestyle{plain}


\section{Introduction}

To meet the 1000x mobile data challenge in the near future~\cite{Qualcomm}, aggressive spectrum reuse and high spectral efficiency must be achieved to significantly boost the capacity of wireless networks. To this end, {\em massive MIMO} (Multiple Input Multiple Output) and {\em small cell} are regarded as two key technologies for emerging 5G wireless systems~\cite{Andrews14}. Massive MIMO (also known as large-scale MIMO or very large MIMO) refers to a wireless system with more than 100 antennas equipped at the base station (BS), which serves multiple users with the same time-frequency resource~\cite{Marzetta10}. Due to the highly efficient spatial multiplexing, a massive MIMO system can achieve dramatically improved energy and spectral efficiency compared to traditional wireless systems~\cite{Ngo13,Xu14Access}. Small cell (or, the notion of {\em network densification}) is another promising approach for capacity enhancement.
With short transmission range and small coverage area, high signal to noise ratio (SNR) and dense spectrum reuse can be achieved, resulting in increased spectral efficiency.

Due to their high potential, the combination of massive MIMO and small cells is expected in future wireless networks, where multiple small cell BS's (SBS) coexist with a macrocell BS (MBS) equipped with a large number of antennas, forming a heterogeneous network (HetNet) with massive MIMO. The two technologies are inherently complementary. On one hand, with a large number of antennas, the MBS has a large number of degrees of freedom (DoF) in the spatial domain, which can be exploited to avoid cross-tier interference.
On the other hand, as the number of users grows, the throughput of a massive MIMO system will be limited by factors such as channel estimation overhead and pilot contamination~\cite{Marzetta10}. By offloading some macrocell users to small cells,
the complexity and overhead of channel estimation at the MBS can be greatly reduced, and the performance of macrocell users can be better guaranteed. Due to these great benefits, massive MIMO HetNet has drawn considerable attention recently~\cite{Hosseini13,Hoydis13,Bjornson13,Bethanabhotla14,Xu15,Liu15}.

However, another advantage of massive MIMO HetNet has not be well considered in the literature, which is its high potential for energy savings.
With the rapid growth of wireless traffic and development of data-intensive services, the power consumption of wireless networks has significantly increased,
which not only generates more ${\rm{C}}{{\rm{O}}_2}$ emission, but also raises the operating expenditure (opex)
of wireless operators. As a result, energy saving, or energy efficiency (EE), becomes a rising concern for the design of wireless networks~\cite{Hasan11,Chen11}. In a recent survey on EE of 5G networks~\cite{Wu15},
it is stated that the design of energy-efficient HetNet with massive MIMO remains a challenge for future study.
A few
schemes have been proposed to improve the EE of massive MIMO HetNets, such as optimizing the beamforming weights of the BS's~\cite{Bjornson13} or optimizing
user association~\cite{Liu15}.

In this paper, we investigate the problem of optimal BS ON-OFF switching and user association in a massive MIMO HetNet, aiming to maximize the EE of the overall system. We propose to dynamically turn ON or OFF SBS's for traffic demands that vary both over time and geographically, as an effective solution to achieve high EE~\cite{Wu15}.
With the high potential for spatial-reuse,
small cells are expected to be densely deployed in the near future, which will result in considerable energy consumption.
As the traffic demand fluctuates over time and space~\cite{Oh11} (e.g., a central business district versus a residential area, and daytime versus nighttime), many of the SBS's will be under-utilized for certain periods of a day, and can be turned OFF to save energy.

A unique advantage of massive MIMO HetNet is that the MBS can provide good coverage for those users initially associated with the turned-OFF SBS's.
However, as
more users are served by the MBS, the MBS performance will become limited by the load-dependent factors, such as channel estimation overhead and pilot contamination. There is clearly a {\em trade-off} here;
the ON-OFF switching strategy of the SBS's should be carefully determined to balance the tension between energy saving and throughput performance. In this paper, we propose a scheme called BOOST (i.e., BS ON-OFF Switching sTrategy) to maximize the EE of a massive MIMO HetNet, by jointly optimizing BS ON-OFF switching and the user association. We develop effective centralized and distributed schemes that can provide optimal solutions to the formulated problem.
%
The main contributions of this paper are summarized as follows.
\begin{itemize}
\item We consider joint BS ON-OFF switching and user association in massive MIMO HetNets, and formulate it as a mixed integer programming problem by taking account of the key design factors.
\item We propose a centralized solution algorithm. In the centralized scheme, we first relax the integer constraints and decompose the relaxed problem into two levels of subproblems. We derive the optimal solution to the relaxed problem with a series of transforms and Lagrangian dual methods, and more important, we prove that this solution is also optimal to the original problem.
\item We propose a distributed scheme based on a user bidding approach. We formulate the bidding process as a repeated game and prove that the game converges to a Nash Equilibrium (NE) that is optimal for each user and BS.
\item The proposed schemes are compared with three benchmark schemes through simulations. The results validate the superior performance of the proposed schemes.
\end{itemize}

In the remainder of this paper, we present the system model and problem formulation in Section~\ref{sec:prob}. The centralized and distributed schemes are presented in Sections~\ref{sec:cent} and~\ref{sec:dist}, respectively.
The simulation results are
discussed in Section~\ref{sec:sim}. We review related work in Section~\ref{sec:related} and conclude this paper in Section~\ref{sec:con}.

\section{Problem Formulation \label{sec:prob}}

The system considered in this paper is based on a noncooperative multi-cell network, and we focus on a tagged macrocell (denoted as macrocell $0$). Macrocell $0$ is a two-tier HetNet consisting of an MBS with a massive MIMO (indexed by $j=0$) and $J$ SBS's (indexed by $j=1,2,\ldots,J$), which collectively serve $K$ mobile users (indexed by $k=1,2,\ldots,K$).
We define binary variables for user association as
\begin{align}
  x_{k,j} \doteq \left\{ \begin{array}{ll}
	         1, & \mbox{user $k$ is connected to BS $j$} \\
					 0, & \mbox{otherwise,}
					          \end{array} \right. \nonumber \\
									k=1,2,\ldots,K, \; j=0,1,\ldots,J.
\end{align}
The MBS is always turned ON to guarantee coverage for users in the macrocell. On the other hand, the SBS's can be dynamically switched ON and OFF for energy savings.
The SBS ON/OFF indicator, denoted as $y_j$, is defined as
\begin{align}
  y_{j} \doteq \left\{ \begin{array}{ll}
	         1, & \mbox{SBS $j$ is turned ON} \\
					 0, & \mbox{SBS $j$ is turned OFF,}
					          \end{array} \right. 
										j=1,2,\ldots,J.
\end{align}

For user $k$ connecting to the MBS in macrocell $0$, let macrocell $l$ be the neighboring macrocell(s) that uses the same pilot sequence as user $k$.
The signal to interference plus noise ratio (SINR) of user $k$, 
$\gamma_{k,0}$,
is given by~\cite{Marzetta10}
\begin{align}\label{eq1}
  \gamma_{k,0} = \frac{\beta_{k,0}^2}{\sum_{l \ne 0} \beta_{k,l}^2},
\end{align}
where $\beta_{k,0}$ is the factor accounts for the propagation loss and shadowing effects between the MBS and user $k$, and $\beta_{k,l}$ is the propagation loss and shadowing factor between user $k$ (the MBS in macrocell $0$) and the MBS in macrocell $l$ (the user in macrocell $l$ that uses the same pilot sequence).

The {\em normalized} achievable data rate of user $k$, when it connects to the MBS in the tagged macrocell, is given as~\cite{Marzetta10}
\begin{align}\label{eq2}
  C_{k,0} = \left( 1 - \sum_{k = 1}^K x_{k,0} \left( \frac{T'}{T} \right) \right) \log \left( 1 + \gamma_{k,0} \right),
\end{align}
where $T$ is the time slot duration and
$T'$ is the time spent to transmit pilot for one user. Note that the number of users that can be served by the MBS at a time is upper bounded~\cite{Marzetta10}, and we denote the upper bound as $S_0$.
We assume that the channel state information (CSI) is collected by the MBS via uplink training (i.e., a time division duplex (TDD) system), so that the MBS can compute $\gamma_{k,0}$ and $C_{k,0}$ for all the users.

We assume that the SBS's use the same block of spectrum as the MBS and adopt frequency division multiple access (FDMA), i.e., the spectrum of SBS $j$, $j=1,2,\ldots,J$, is divided into $S_j$ channels and each of its user is allocated with one channel. Thus, the number of users that can be served by SBS $j$ at a time is upper bounded by $S_j$. We also assume that both the SBS's and users transmit with a fixed power, and the SINR $\gamma_{k,j}$ can be measured by the SBS via uplink training. Thus, for a user $k$ connecting to SBS $j$, the {\em normalized} achievable data rate of the user can be written as
\begin{align}\label{eq3}
  C_{k,j} = \frac{\log \left( 1 + \gamma_{k,j} \right)}{S_j}  =  \frac{R_{k,j}}{S_j}, \; j=1, 2, \ldots, J,
\end{align}
where ${R_{k,j}} = \log \left( 1 + \gamma_{k,j} \right)$, $j=0,1,\ldots,J$.

The power consumption model of HetNets is investigated in~\cite{Arnold10}. The power consumption of a BS consists of a static part and a dynamic part. The static part is the power required for the operation of a BS once it is turned ON, e.g., used by the cooling system, power amplifier, and baseband units. The dynamic part is dependent on the traffic load, e.g., used by the radio frequency unit. The study in~\cite{Arnold10} shows that the static part constitutes a dominant proportion of the total BS power consumption, especially for the MBS. Furthermore, the power consumption when a BS is fully-loaded is close to
that when the BS is under-loaded. Thus, we approximate the power consumption of each BS as
$P_j$, $j=0,1,\ldots,J$. The total power consumption of the HetNet is ${P_0} + \sum_{j = 1}^J y_j P_j$.

In this paper, we aim to maximize the energy efficiency of a HetNet with massive MIMO.
Let $\bf{x}$ and $\bf{y}$ denote the $\{ x_{k,j} \}$ matrix and the $\{ y_j \}$ vector, respectively.
The problem is formulated as
\begin{align}
& {\bf P1:\/}  \max_{\left\{ {\bf{x}}, {\bf{y}} \right\}} \frac{\sum_{k = 1}^K x_{k,0} C_{k,0} + \sum_{k = 1}^K \sum_{j = 1}^J x_{k,j} C_{k,j}}{ P_0 + \sum_{j = 1}^J y_j P_j } \label{eq4} \\
&\mbox{subject to:} \nonumber \\
&\hspace{0.3in} \sum_{j = 0}^J x_{k,j} \le 1, \; k=1,2,\ldots,K \label{eq5} \\
&\hspace{0.3in} \sum_{k = 1}^K x_{k,j} \le {S_j}, \; j=0,1,\ldots,J \label{eq6} \\
&\hspace{0.3in} \; x_{k,j} \le y_j, \; k=1,2,\ldots,K, \; j=1,2,\ldots,J \label{eq7} \\
&\hspace{0.3in} \; x_{k,j} \in \left\{ 0, 1 \right\}, \; k=1,2,\ldots,K, \; j=0,1,\ldots,J \label{eq8} \\
&\hspace{0.3in} \; {y_j} \in \left\{ 0, 1 \right\}, \; j=1,2,\ldots,J. \label{eq88}
\end{align}
In problem {\bf P1}, constraint~(\ref{eq5}) is due to fact that each user can connect to at most one BS; constraint~(\ref{eq6}) enforces the
upper bound for the number of users that can be served by BS $j$; and constraint~(\ref{eq7}) is because users can connect to SBS $j$ only when it is turned ON (i.e., when $y_j=1$).

\section{Centralized Solution Algorithm \label{sec:cent}}

In general the small cells are deployed by the operator and can use the X2 interface to communicate with each other as well as the MBS. A centralized algorithm can be useful in this context to coordinate their operations.
In this section, we solve the formulated problem with a centralized scheme and prove the optimality of the derived solution.

Problem {\bf P1\/} is an integer programming problem, which is generally NP-hard. To develop an effective solution algorithm, we relax the integer constraints by allowing $x_{k,j}$ and $y_j$ to take values in $\left[ 0, 1 \right]$. Thus, {\bf P1\/} is transformed into
\begin{align}
&{\bf P2:\/} \max_{\left\{ {\bf{x}}, {\bf{y}} \right\}} \frac{\sum_{k = 1}^K x_{k,0} C_{k,0} + \sum_{k = 1}^K \sum_{j = 1}^J x_{k,j} C_{k,j}}{ P_0 + \sum_{j = 1}^J y_j P_j } \label{eq9} \\
&\mbox{subject to:} \nonumber \\
&\hspace{0.3in} \sum\limits_{j = 0}^J {{x_{k,j}}}  \le 1, \; \forall \; k \label{eq10} \\
&\hspace{0.3in} \sum\limits_{k = 1}^K {{x_{k,j}}}  \le {S_j}, \; \forall \; k,j \label{eq11} \\
&\hspace{0.3in} \; {x_{k,j}} \le {y_j}, \; \forall \; k,j \label{eq12} \\
&\hspace{0.3in} \; 0 \le {x_{k,j}},{y_j} \le 1, \; \forall \; k,j. \label{eq13}
\end{align}

The decision variables $x_{k,j}$ and $y_j$ are coupled in the constraints, which are difficult to handle directly.
To solve the relaxed problem, we decompose problem {\bf P2\/} into two levels of subproblems. At the lower-level subproblem, we find the optimal solution for $\bf{x}$
for given values of $\bf{y}$.
Based on the solution at the lower-level subproblem, we obtain the optimal value of $\bf{y}$
at the higher-level subproblem through a subgradient approach.

\subsection{Optimal Solution of Problem {\bf P2\/} for a Given $\bf{y}$}

For given values of $\bf{y}$, the lower-level subproblem of problem {\bf P2\/} becomes the following problem ${\bf P3}$.
\begin{align}
& {\bf P3:\/}  \max_{\{ {\bf{x}} \}} \sum_{k = 1}^K x_{k,0} C_{k,0} + \sum_{k = 1}^K \sum_{j = 1}^J x_{k,j} C_{k,j} \label{eq14} \\
&\mbox{subject to:} \nonumber \\
&\hspace{0.3in} \sum\limits_{j = 0}^J {{x_{k,j}}}  \le 1, \; \forall \; k \label{eq15} \\
&\hspace{0.3in} \sum\limits_{k = 1}^K {{x_{k,j}}}  \le {S_j}, \; \forall \; k,j \label{eq16} \\
&\hspace{0.3in} \; {x_{k,j}} \le {y_j}, \; \forall \; k,j \label{eq17} \\
&\hspace{0.3in} \; 0 \le {x_{k,j}} \le 1, \; \forall \; k,j. \label{eq18}
\end{align}

We first show that problem {\bf P3\/} is a convex problem that can be solved in the dual domain.

\begin{lemma}\label{lemma1}
Problem {\bf P3\/} is a convex optimization problem.
\end{lemma}

\begin{proof}
Since the constraints of problem {\bf P3\/} are linear, we only need to show that
the objective function is concave.
Substituting~(\ref{eq2}), the objective function
is a combination of two parts: a linear function of $\bf{x}$ and the following term
\begin{align}\label{eq19}
  E \doteq - \frac{{{T'}}}{T}\left( {\sum\limits_{k = 1}^K {{x_{k,0}}} } \right)\left( {\sum\limits_{k = 1}^K {{x_{k,0}}{R_{k,0}}} } \right).
\end{align}

The Hessian
of $E$ is given by
\begin{align} 
 & {\bf{H}}_{K \times K} \nonumber \\
=& - \frac{T'}{T} \left( \hspace{-0.025in} {\begin{array}{*{20}{c}}
        {2{R_{1,0}}}&{{R_{1,0}} \hspace{-0.025in}+\hspace{-0.025in} {R_{2,0}}}& \cdots &{{R_{1,0}} \hspace{-0.025in}+\hspace{-0.025in} {R_{K,0}}}\\
        {{R_{1,0}} \hspace{-0.025in}+\hspace{-0.025in} {R_{2,0}}}&{2{R_{2,0}}}& \cdots &{{R_{2,0}} \hspace{-0.025in}+\hspace{-0.025in} {R_{K,0}}}\\
        \vdots & \vdots & \ddots & \vdots \\
        {{R_{1,0}} \hspace{-0.025in}+\hspace{-0.025in} {R_{K,0}}}&{{R_{2,0}} \hspace{-0.025in}+\hspace{-0.025in} {R_{K,0}}}& \cdots &{2{R_{K,0}}}
                             \end{array}} \hspace{-0.025in} \right). \nonumber
\end{align}
Let ${\bf{z}} = {\left[ {{z_1},{z_2},\ldots,{z_k}} \right]^T}$ be an arbitrary non-zero vector. We have
\begin{align}
& {\bf{z}}^T  \bf{Hz} \nonumber \\
=& - \frac{{{2T'}}}{T}\left[ {\sum\limits_{k = 1}^K {z_k^2{R_{k,0}}}  + \sum\limits_{k = 1}^K {\sum\limits_{{k^{'}} \ne k} {{z_k}{z_{{k^{'}}}}\left( {{R_{k,0}} + {R_{{k^{'}},0}}} \right)} } } \right] \nonumber \\
\mathop  < \limits^{\left( a \right)}& - \frac{{{2T'}}}{T}\left[ {\sum\limits_{k = 1}^K {z_k^2{R_{k,0}}}  + \sum\limits_{k = 1}^K {\sum\limits_{{k^{'}} \ne k} {{z_k}{z_{{k^{'}}}}\left( {2\sqrt {{R_{k,0}}{R_{{k^{'}},0}}} } \right)} } } \right] \nonumber \\
=& - \frac{{{2T'}}}{T}{\left( {\sum\limits_{k = 1}^K {{z_k}\sqrt {{R_{k,0}}} } } \right)^2} < 0. \label{eq21}
\end{align}
Inequality $(a)$ results from the fact that for two positive numbers, $m + n \ge 2\sqrt{mn}$ and the equality holds when $m=n$.

From~(\ref{eq21}), we conclude that $E$ given in~\eqref{eq19}
is a concave function. The objective function of problem {\bf P3\/} is thus concave, since the sum of concave functions is still concave. Therefore problem {\bf P3\/} is a convex optimization problem.
\end{proof}
\smallskip

To obtain the optimal solution of problem {\bf P3\/}, we introduce an auxiliary variable ${Q_0} \doteq \sum_{k = 1}^K {x_{k,0}}$. To deal with the coupling variables $Q_0$ and $x_{k,0}$, we further decompose problem {\bf P3\/} into two levels of subproblems. At the lower-level, we find the optimal solution for ${\bf{x}}$ for given $Q_0$. At the higher-level, the optimal value of $Q_0$ is obtained based on the solution of the lower-level subproblem.

\subsubsection{Lower-level Subproblem of Problem {\bf P3\/}}

For a given $Q_0$, problem {\bf P3\/} can be transformed into the following problem {\bf P4\/}.
\begin{align}
& {\bf P4:\/} \max_{\bf{x}} \sum_{k = 1}^K \left( \sum_{j = 1}^J x_{k,j} \frac{R_{k,j}}{S_j} + x_{k,0} R_{k,0} - x_{k,0} R_{k,0} \frac{Q_0 T'}{T} \right)  \label{eq22} \\
&\mbox{subject to:} \nonumber \\
&\hspace{0.3in} \sum\limits_{j = 0}^J {{x_{k,j}}}  \le 1, \; \forall \; k \label{eq23} \\
&\hspace{0.3in} \sum\limits_{k = 1}^K {{x_{k,j}}}  \le {S_j}, \; \forall \; k,j \label{eq24} \\
&\hspace{0.3in} {x_{k,j}} \le {y_j}, \; \forall \; k,j \label{eq25} \\
&\hspace{0.3in} \sum\limits_{k = 1}^K {{x_{k,0}}} = {Q_0} \label{eq26} \\
&\hspace{0.3in} 0 \le {x_{k,j}} \le 1, \; \forall \; k,j. \label{eq27}
\end{align}

Relax the constraints on $Q_0$ and $y_j$, the dual problem of {\bf P4\/} is given by
\begin{align}\label{eq28}
& \mbox{\bf P4-Dual:\/} \;\; \min_{\{ {\boldsymbol{\lambda}}, \mu\}} \; g ( {\boldsymbol{\lambda}}, \mu ),
\end{align}
where $\boldsymbol{\lambda }$ and $\mu$ are the Lagrangian multipliers for
constraints~\eqref{eq25} and~\eqref{eq26}, respectively; and
$g ( {\boldsymbol{\lambda}}, \mu )$ is given by
\begin{align}
&\; g ( {\boldsymbol{\lambda}}, \mu ) \nonumber \\
=&\; \max_{\{ \bf{x} \}} \sum_{k = 1}^K \left( \sum_{j = 1}^J x_{k,j} \frac{R_{k,j}}{S_j} \hspace{-0.025in}+\hspace{-0.025in} x_{k,0} R_{k,0} \hspace{-0.025in}-\hspace{-0.025in} x_{k,0} R_{k,0} \frac{Q_0 T'}{T} \right) + \nonumber \\
&\;  \sum_{k = 1}^K \sum_{j = 1}^J \lambda _{k,j} \left( y_j - x_{k,j} \right) +
\mu \left( Q_0 - \sum_{k = 1}^K x_{k,0} \right). \label{eq29}
\end{align}

The optimal solution of the dual problem can be obtained with the following subgradient method.
\begin{align}
\left\{ \begin{array}{ll}
   \lambda_{k,j}^{[ t + 1 ]} = \left[ \lambda_{k,j}^{[t]} + \tau \left( x_{k,j}^{[t]} - y_j^{[t]} \right) \right]^+, & \forall \; k,j  \\
   \mu^{[t+1]} = \mu^{[t]} + \tau \left( Q_0^{[t]} - \sum_{k = 1}^K x_{k,0}^{[t]} \right),
        \end{array} \right.  \label{eq30}
\end{align}
where $\tau$ is the step size for each iteration, $\left[ z \right]^+ \doteq \max \left\{ 0, z \right\}$, and $t$ is the index of iterations.

Given $\lambda_{k,j}$, $\mu$, and $Q_0$, the maximization of~(\ref{eq29}) is a standard
LP given by
\begin{align}
  {\bf P5:\/} & \;\;\max_{\{ \bf{x} \}}g\left( {\boldsymbol{\lambda}},\mu \right) \label{eq31} \\
  \mbox{subject to:} & \;\; (\ref{eq23}), (\ref{eq24}), \mbox{and } (\ref{eq27}). \nonumber
\end{align}
Problem {\bf P5\/} can be solved with effective methods for LPs, such as the simplex algorithm. We next show that the solution variables of problem {\bf P5\/} are integers rather than fractions in $(0,1)$, although with the relaxed constraint~(\ref{eq27}).

\begin{definition}\label{definition1}
A matrix ${\bf{A}}$ is totally unimodular if the determinant of every square submatrix of ${\bf{A}}$ is either 0, +1 or -1~\cite{Schrijver98}.
\end{definition}

\begin{property}\label{property1}
If the constraint matrix of an LP satisfies totally unimodularity, then it has all integral vertex solutions~\cite{Schrijver98}.
\end{property}

\begin{property}\label{property2}
If an LP has feasible optimal solutions, then at least one of the feasible optimal solutions occurs at a vertex of the polyhedron defined by its constraints~\cite{Berenstein97}.
\end{property}

To analyze the properties of problem {\bf P5\/}, we define a new vector ${\bf{\tilde x}}$ by concatenating all the columns of ${\bf{x}}$ as ${\bf{\tilde x}} = \left[ x_{1,1}, \ldots, x_{K,1}, x_{1,2}, \ldots, x_{K,2}, \ldots, \ldots, x_{1,J}, \ldots, x_{K,J} \right]^T$. We then rewrite problem {\bf P5\/} in the following standard form.
\begin{align}
  \max_{\bf{\tilde x}} & \;\; {\bf{c}} {\bf{\tilde x}} \\
  \mbox{subject to:}   & \;\; {\bf{A}} {\bf{\tilde x}} \le {\bf{b}},
\end{align}
where ${\bf{c}}$ is the vector of coefficients for~(\ref{eq31}), and the constraint matrix $\bf A$ and vector $\bf b$ are given, respectively, by
\begin{align}
& {\bf{A}} \doteq \left( {\begin{array}{*{20}{c}}
{1{\rm{~}}1{\rm{~}} \cdot  \cdot  \cdot {\rm{~}}1}&{0{\rm{~}}0{\rm{~}} \cdot  \cdot  \cdot {\rm{~}}0}&{ \cdot  \cdot  \cdot }&{0{\rm{~}}0{\rm{~}} \cdot  \cdot  \cdot {\rm{~}}0}\\
{0{\rm{~}}0{\rm{~}} \cdot  \cdot  \cdot {\rm{~}}0}&{1{\rm{~}}1{\rm{~}} \cdot  \cdot  \cdot {\rm{~}}1}&{ \cdot  \cdot  \cdot }&{0{\rm{~}}0{\rm{~}} \cdot  \cdot  \cdot {\rm{~}}0}\\
 \vdots & \vdots &{}& \vdots \\
{0{\rm{~}}0{\rm{~}} \cdot  \cdot  \cdot {\rm{~}}0}&{0{\rm{~}}0{\rm{~}} \cdot  \cdot  \cdot {\rm{~}}0}& \cdots &{1{\rm{~}}1{\rm{~}} \cdot  \cdot  \cdot {\rm{~ }}1}\\
{1{\rm{~0~}} \cdot  \cdot  \cdot {\rm{~0}}}&{1{\rm{~0~}} \cdot  \cdot  \cdot {\rm{~0}}}& \cdots &{1{\rm{~0~}} \cdot  \cdot  \cdot {\rm{~0}}}\\
{{\rm{0~}}1{\rm{~}} \cdot  \cdot  \cdot {\rm{~0}}}&{{\rm{0~}}1{\rm{~}} \cdot  \cdot  \cdot {\rm{~0}}}& \cdots &{{\rm{0~}}1{\rm{~}} \cdot  \cdot  \cdot {\rm{~ 0}}}\\
 \ddots & \ddots & \vdots & \ddots \\
{{\rm{0~0~}} \cdot  \cdot  \cdot {\rm{~}}1}&{{\rm{0~0~}} \cdot  \cdot  \cdot {\rm{~}}1}& \cdots &{{\rm{0~0~}} \cdot  \cdot  \cdot {\rm{~}}1}
\end{array}} \right) \label{eq35} \\
& {\bf{b}}_{\left( {K + J} \right) \times 1} \doteq \left[ {{S_1},\ldots,{S_J},1,\ldots,1} \right]^T.
\end{align}

\begin{lemma}\label{lemma2}
The constraint matrix ${\bf{A}}$ is totally unimodular.
\end{lemma}

\begin{proof}
We divide the constraint matrix $\bf{A}$ into blocks as
\begin{align}\label{eq36}
{\bf{A}} = \left( {\begin{array}{*{20}{c}}
{{{\bf{W}}_1}}&{{{\bf{W}}_2}}& \cdots &{{{\bf{W}}_J}}\\
{{{\bf{I}}_1}}&{{{\bf{I}}_2}}& \cdots &{{{\bf{I}}_J}}
\end{array}} \right),
\end{align}
where each ${\bf{W}}_j$, $j=1,2,\ldots,J$, is a $J \times K$ matrix; the $j$th row of ${\bf{W}}_j$ is all $1$, while all the other rows are all $0$; and each ${\bf{I}}_j$, $j=1,2,\ldots,J$, is a $K \times K$ identity matrix.

Denote $G_n$ as an arbitrary $n \times n$ square submatrix of matrix ${\bf A}$. Obviously, the determinant of $G_n$ is either $0$ or $1$ when $n=1$. To analyze the determinant of $G_n$ for $n \geq 2$, the following two cases need to be considered.

\noindent {\bf{Case 1}}: $G_n$ is a submatrix of ${\bf{W}}_j$ or ${\bf{I}}_j$, $j=1,2,\ldots,J$. If $G_n$ is a submatrix of ${\bf{W}}_j$, we have $\det \left( G_n \right) = 0$, since at least one row would be all $0$. If $G_n$ is a submatrix of ${{{\bf{I}}_j}}$, $\det \left( {{G_n}} \right)$ would be either $0$ or $+1$, since ${{{\bf{I}}_j}}$ is an identity matrix.

\noindent {\bf{Case 2}}: The entries of $G_n$ are from more than one ${{{\bf{W}}_j}}$ or ${{{\bf{I}}_j}}$. We apply an induction method to analyze the determinant.
For $n=2$, $\det \left( {{G_n}} \right)$ can only be $0$, $+1$, or $-1$, since the four entries are either $0$ or $1$ with at least one $0$. Suppose $\det \left( {{G_{n-1}}} \right)$ can only be $0$, $+1$, or $-1$, we need to verify whether the same result hold for $\det \left( {{G_n}} \right)$. Denote ${G_n}\left( {u,v} \right)$ as the entry of $G_n$ at row $u$, column $v$. Let $v^* = \argmin_v \{ \sum_u G_n \left( u,v \right) \}$. Then column $v^*$ is the one with the minimum number of $1$s in $G_n$. Let $\Delta_{v^*}$ be the number of $1$s in column $v^*$,
which can be $0$, $1$, or $2$ according to the structure of ${\bf A}$ shown in~\eqref{eq35}.

    If $\Delta_{v^*}=0$, column $v^*$ of $G_n$ is all $0$ and $\det \left( G_n \right) = 0$.

    If $\Delta_{v^*}=1$, we calculate $\det \left( G_n \right)$ through column $v^*$ and have $\det \left( G_n \right) = \pm \det \left( G_{n-1} \right)$. According to the induction hypothesis, $\det \left( G_{n-1} \right)$ can only be $0$, $-1$, or $1$. Therefore, $\det \left( G_n \right)$ can only be $0$, $-1$, or $1$.

    If $\Delta_{v^*}=2$, each column of $G_n$ has exactly two $1$s, with one in ${\bf{W}}_j$ and the other in ${\bf{I}}_j$. Due to the equal number of $1$s in ${{{\bf{W}}_j}}$ and ${{{\bf{I}}_j}}$, we can obtain an all-zero row in $G_n$ through some elementary transformations, which yields $\det \left( G_n \right)=0$.

Consequently, the determinant of any square submatrix of ${\bf{A}}$ can only be either $0$, $-1$, or $1$. According to Definition~\ref{definition1}, we conclude that $\bf{A}$ is totally unimodular.
\end{proof}

\begin{lemma}\label{lemma3}
All the decision variables in the optimal solution to the relaxed LP, problem {\bf P5\/}, are integers in $\left\{ 0, 1 \right\}$.
\end{lemma}
\begin{proof}
This lemma directly follows Lemma~\ref{lemma2}, Property~\ref{property1}, and Property~\ref{property2}.
\end{proof}

\medskip
\subsubsection{Higher-level Subproblem of Problem {\bf P3\/}}

We first show that the {\em duality gap} between the lower level subproblem {\bf P4\/} and its dual, problem {\bf P4-Dual\/}, is zero.

\begin{lemma}\label{lemma4}
Strong duality holds for problem {\bf P4\/}.
\end{lemma}

\begin{proof}
Since problem {\bf P4\/} is an LP, all the constraints are linear and the Slater condition reduces to feasibility~\cite{Boyd04}. Thus strong duality holds.
\end{proof}
\smallskip

Let $f({\bf{x}})$ be the objective function of problem {\bf P4\/} for a given ${\bf{x}}$. In the higher-level subproblem of problem {\bf P3\/}, we find the optimal value of $Q_0$ by solving the following problem.
\begin{align}\label{eq37}
{\bf P6:\/} \;\; \max_{\left\{ Q_0 \right\}} \; f ({\bf{x}} (Q_0)).
\end{align}

\begin{lemma}\label{lemma5}
Problem {\bf P6\/} can be solved with the following subgradient method.
\begin{align}\label{eq38}
   Q_0^{[t+1]} = Q_0^{[t]} + \tau \cdot \mu^{[t]}.
\end{align}
\end{lemma}

\begin{proof}
Let ${\bf{x}}^* (Q_0^{'})$ be the optimal solution to problem {\bf P4\/} for a given value of $Q_0^{'} = \sum_{k = 1}^K x_{k,0}^*$, and $f^{*} (Q_0^{'})$ be the optimal objective value
with solution ${\bf{x}}^* (Q_0^{'})$. Denote $\mathcal{L}(\cdot)$ as the {\em Lagrangian function}. For another feasible solution $\bf{x}$ to problem {\bf P4\/} with a given value $Q_0 = \sum_{k = 1}^K x_{k,0}$, the following equalities and inequalities hold.
\begin{align}
& \;\; f^*(Q_0^{'})
\mathop = \limits^{(a)} \mathcal{L} \left( {{\bf{x}}^*}, {\boldsymbol{\lambda}}^* ( Q_0^{'} ), {\mu^*} ( Q_0^{'} ) \right) \nonumber \\
\mathop  \ge \limits^{\left( b \right)} & \mathcal{L} \left( {\bf{x}}, {\boldsymbol{\lambda}}^* ( Q_0^{'} ), {\mu^*} ( Q_0^{'} ) \right) \nonumber \\
= &\; f({\bf{x}}) + \sum_{k = 1}^K \sum_{j = 1}^J \lambda_{k,j}^* ( y_j - x_{k,j} ) + {\mu^*} \left( Q_0^{'} - \sum_{k = 1}^K x_{k,0} \right) \nonumber \\
= &\; f({\bf{x}}) + \sum_{k = 1}^K \sum_{j = 1}^J \lambda_{k,j}^* ( y_j - x_{k,j} ) + \mu^* \left( Q_0 - \sum_{k = 1}^K x_{k,0} \right) + \nonumber \\
  &\; {\mu^*}( Q_0^{'} - Q_0 ) \nonumber \\
\mathop  \ge \limits^{(c)} &\; f\left( {\bf{x}} \right) + {\mu^*}\left( Q_0^{'} - Q_0 \right), \label{eq39}
\end{align}
where equality $(a)$ is due to strong duality, inequality $(b)$ is due to the optimality of ${{\bf{x}}^*}$, and inequality $(c)$ is due to the constraints of problem {\bf P4\/} and the nonnegativity of ${\boldsymbol{\lambda}}$. Note that $(c)$ holds for any ${\bf{x}}$ such that $\sum_{k = 1}^K x_{k,0} = Q_0$.

In particular, we have
\begin{align} 
f^* (Q_0^{'}) \ge&\; \max_{\{{\bf{x}}|\sum_{k = 1}^K x_{k,0} = Q_0 \}} \left\{ f ({\bf{x}}) + \mu^* ( Q_0^{'} - Q_0) \right\} \nonumber \\
 =&\; f^* ( Q_0 ) + \mu^* (Q_0^{'}) ( Q_0^{'} - Q_0 ). \label{eq40}
\end{align}
It follows~\eqref{eq40} that
\begin{align} 
f^* (Q_0) \le f^* ( Q_0^{'} ) + {\mu^*} ( Q_0^{'} ) (Q_0 - Q_0^{'} ). \nonumber
\end{align}
By definition, ${\mu^*} (Q_0^{'})$ is a subgradient of $f^* ( Q_0 )$. Thus, we conclude that problem {\bf P6\/} can be solved with~(\ref{eq38}).
\end{proof}

\medskip

\subsection{Optimal Value of ${\bf{y}}$}

Let $D^*(\bf{y})$ be the optimal value of problem {\bf P2\/} for given ${\bf{y}}$. The optimal ${\bf{y}}$ can be
solved from the following problem.
\begin{align}\label{eq42}
   {\bf P7:\/} \;\; \max_{\{ {\bf{y}} \}} \; D^* ( {\bf{y}} ).
\end{align}

\begin{lemma}\label{lemma6}
Problem {\bf P7\/} can be solved by the following subgradient method.
\begin{align}\label{eq43}
{\bf{y}}^{[t+1]} = {\bf{y}}^{[t]} + \tau \cdot {\boldsymbol{\nu}}^{[t]},
\end{align}
where ${\boldsymbol{\nu}} = \left[ \sum_{k = 1}^K \lambda_{k,1}^*, \sum_{k = 1}^K \lambda_{k,2}^*, \ldots, \sum_{k = 1}^K \lambda_{k,J}^* \right]$.
\end{lemma}

\begin{proof}
Let ${\bf{x}}^* ({\bf{y}}')$ be the optimal solution to problem {\bf P2\/} for a given ${\bf{y}}'$ and $\bf{x}$ be another feasible solution for given $\bf{y}$. Similar to the analysis for problem {\bf P6\/} in~(\ref{eq39}), the following equalities and inequalities hold.
\begin{align} 
&\; D^* ( {\bf{y}}' ) = D ( {\bf{x}^*} ( {\bf{y}}' ) ) = \mathcal{L} ( {\bf{x}}^*, {\boldsymbol{\lambda}}^* ( {\bf{y}}' ) ) \nonumber \\
 \ge&\; \mathcal{L} ( {\bf{x}}, {\boldsymbol{\lambda}}^* ( {\bf{y}}' ) ) = D ( {\bf{x}} ) + \sum_{k = 1}^K {\boldsymbol{\lambda}}_k^* ( {\bf{y}}' )  ( {\bf{y}}' - {\bf{x}}_k ) \nonumber \\
 =&\; D ( {\bf{x}} ) + \sum_{k = 1}^K {\boldsymbol{\lambda}}_k^* ( {\bf{y}}' ) ( {\bf{y}} - {\bf{x}}_k ) + \sum_{k = 1}^K {\boldsymbol{\lambda}}_k^* ( {\bf{y}}' ) ( {\bf{y}}' - {\bf{y}} ) \nonumber \\
 \ge&\; D ( {\bf{x}} ) + \sum_{k = 1}^K {\boldsymbol{\lambda}}_k^* ( {\bf{y}}' ) ( {\bf{y}}' - {\bf{y}} ) = D ( {\bf{x}} ) + {\boldsymbol{\nu}} ( {\bf{y}}' - {\bf{y}} ), \nonumber
\end{align}
where ${\bf{x}}_k = \left[ x_{k,1}, x_{k,2}, \ldots, x_{k,J} \right]^T$ and ${\boldsymbol{\lambda}}_k^* ( {\bf{y}}' )$ is the $k$th row of ${\boldsymbol{\lambda}}^* ( {\bf{y}}' )$.

In particular, we have
\begin{align} 
D^* ({\bf{y}}') \ge \max_{\{ {\bf{x}} \le {\bf{y}} \}} \left\{ D ({\bf{x}}) + {\boldsymbol{\nu}} ( {\bf{y}}' \hspace{-0.025in}-\hspace{-0.025in} {\bf{y}} ) \right\} = D^* ( {\bf{y}} ) + {\boldsymbol{\nu}} ({\bf{y}}' \hspace{-0.025in}-\hspace{-0.025in} {\bf{y}}). \nonumber
\end{align}
Thus, ${\boldsymbol{\nu}}$ is a subgradient of ${\bf{y}}$. We conclude that problem {\bf P7\/} can be solved by~(\ref{eq43}).
\end{proof}

\subsection{Optimality Analysis}

The procedure of the centralized user association and BS ON-OFF switching strategy is summarized in Algorithm~\ref{alg:1}.
\begin{algorithm} [!t]
\small
\SetKwRepeat{doWhile}{do}{while}
\SetAlgoLined
	Initialize ${Q_0},{\bf{y}},{\boldsymbol{\lambda }}$, and ${\boldsymbol{\mu }}$ \;
\doWhile{$($${\bf{y}}$ does not converge$)$}{
 \doWhile{$($$Q_0$ does not converge$)$}{
  \doWhile{$($${\boldsymbol{\lambda }},{\boldsymbol{\mu }}$ do not converge$)$}{
			Solve problem {\bf P5\/} with a standard LP solver \;
			Update ${\boldsymbol{\lambda}},{\boldsymbol{\mu}}$ as in~(\ref{eq30}) \;
  }
  Update $Q_0$ as in~(\ref{eq38}) \;
  }
  Update ${\bf{y}}$ as in~(\ref{eq43})\;
}
\caption{Centralized User Association and BS ON-OFF Switching Strategy}
\label{alg:1}
\end{algorithm}

\begin{theorem}
The solution produced by Algorithm~\ref{alg:1} is optimal to problem {\bf P1\/}.
\end{theorem}

\begin{proof}
According to Lemma~\ref{lemma3}, the variables in the optimal solution of problem {\bf P5\/} are all binary
for any feasible values of ${\boldsymbol{\lambda}}$, ${\boldsymbol{\mu}}$, and $Q_0$. With
Algorithm~\ref{alg:1}, the optimal solution for ${\bf{x}}$ in problem {\bf P2\/} is also a binary vector.

Consider constraint~\eqref{eq12}, i.e., $x_{k,j} \le y_j, \forall k,j$. If there exists a ${x_{k,j}} = 1$ for any $k$, then ${y_j}$ must be equal to $1$ to satisfy the constraint. If ${x_{k,j}} = 0$ for all $k$, then the constraint~\eqref{eq12}
is always satisfied; the optimal value for ${y_j}$ must be $0$ since the objective value with ${y_j}=0$ is larger than that with ${y_j}=1$. Thus, the optimal solution for ${\bf{y}}$ in problem {\bf P2\/} is also a binary vector.
Since the optimal solution of problem {\bf P2\/} is binary,
the solution is also feasible and optimal for problem {\bf P1\/}. We conclude that Algorithm~\ref{alg:1}
is optimal for problem {\bf P1\/}.
\end{proof}

\section{Distributed Solution Scheme \label{sec:dist}}

In the previous section, we proposed a centralized scheme that is optimal but requires global network information. However, centralized control may not always be feasible due to constraints on complexity, overhead, or scalability.
In this section, we propose a distributed scheme based on a user bidding approach. The bidding procedure is formulated as a repeated game between users and BS's, and we demonstrate that the game will converge to a Nash Equilibrium (NE).

\subsection{Distributed User Association and SBS ON-OFF Control}

We assume that the utility of each user $k$ is positively correlated to the achievable rate $C_{k,j}$ and user $k$ always seeks to maximize $C_{k,j}$. The {\em preference list} of user $k$ is determined by the $C_{k,j}$ values for different $j$. For instance, if $j^* = \argmax_j \{ C_{k,j} \}$, BS $j^*$ is on top of user $k$'s preference list. The preference list of BS $j$ is also determined by $C_{k,j}$ is a similar way. Denote the price paid by user $k$ to BS $j$ as $p_{k,j}$. It is reasonable to assume that $p_{k,j}$ is proportional to $C_{k,j}$. The utility of BS $j$ is defined as the payments made by all its connected users subtract the cost of power consumption $q_j$, given by
\begin{align}\label{eq46}
   \sum_{k = 1}^K x_{k,j} p_{k,j} - {q_j}.
\end{align}

\begin{algorithm} [!t]
\small
\SetAlgoLined
	  \While{$($convergence not achieved$)$}{
				\eIf{$($more than $S_j$ users bid for BS $j$$)$}{
					 Put the top $S_j$ users with the highest bids in the waiting list and reject the other users \;
				}{
					Put all users in the waiting list \;
		 }
	 }
\caption{Distributed User Association Strategy for BS's}
\label{alg:2}
\end{algorithm}

To maximize the total utility under the constraint $\sum_{k = 1}^K x_{k,j} = S_j$, the distributed user association strategy for the BS's is presented in Algorithm~\ref{alg:2}. Note that the values of $p_{k,j}$ and $q_j$ can be optimized to further enhance the performance, we omit the analysis in this paper due to the page limit.
The repeated bidding game has two stages. In the {\em first} stage, each user bids for the top BS in its preference list. After receiving the bids, the MBS and SBS's decide the user association strategy according to Algorithm~\ref{alg:2} and feedback the decision to users.

In the {\em second} stage, if a user has been rejected, the BS that rejected it would be deleted from its preference list. Then, the user bids for the most desirable BS among the remaining ones. Upon receiving the bids, each BS compares the new bids with those in its waiting list, and makes decisions on user association according to Algorithm~\ref{alg:2}. The rejected users then make another round of bids following the order of their preference lists, and the BS's again make decisions and feedback to users, and so forth. The bidding procedure is continued until convergence is achieved, i.e., the users in the waiting list of each BS do not change anymore. 

After convergence of the user association result, each SBS determines the value of its ON-OFF decision variable by comparing the payments and energy cost as follows.
\begin{align}\label{eq47}
y_j = \left\{ {\begin{array}{ll}
				1, & \mbox{if} \; \sum_{k = 1}^K x_{k,j} p_{k,j} > {q_j} \\
				0, & \mbox{otherwise,}
\end{array}} \right. \; j = 1,2,\ldots,J.
\end{align}
It can be seen from~(\ref{eq47}) that SBS $j$ chooses to be turned ON only when it is profitable to do so. The users in the waiting list of SBS $j$
will connect to the MBS if SBS $j$ is turned OFF.

\subsection{Convergence Analysis}

We next prove that the repeated game converges and an NE
can be achieved.

\begin{lemma}\label{lemma7}
The sequence of bids made by a user is non-increasing in its preference list.
\end{lemma}

\begin{proof}
To maximize utility, a user first bids for the most desirable BS in its preference list. If rejected, the user deletes the BS from its preference list and bids for the most desirable BS in the updated list. Thus, the BS's chosen by a user is non-increasing in its preference list.
\end{proof}

\begin{lemma}\label{lemma8}
The sequence of bids in the waiting list of a BS is non-decreasing in its preference list.
\end{lemma}

\begin{proof}
According to Algorithm~\ref{alg:2}, if a BS is not fully loaded, it put all the bids into the waiting list. If a BS is fully loaded, it compares the new incoming bids with the bids already in the list, and selects the most profitable bids to maximize its own utility.
\end{proof}

\begin{theorem}\label{theorem2}
The repeated game converges.
\end{theorem}

\begin{proof}
Suppose the game does not converge. Then, there must be a user $k$ and a BS $j$ such that: (i) user $k$ prefers BS $j$ to its current connecting BS ${j'}$, (ii) BS $j$ prefers user $k$ to user $k'$, who is currently in the waiting list of BS $j$. Under this circumstance, user $k$ is a better choice and BS $j$ can accept the bid of user $k$. User $k$ will bid for BS $j$.

Based on Lemma~\ref{lemma8}, the sequence of bids received by BS $j$ is non-decreasing. As user $k$ is a better choice than user $k'$ for BS $j$ while user $k$ is not in the waiting list, it must be the case that user $k$ has never bidden for BS $j$. Since user $k$ prefers BS $j$ to BS $j'$, user $k$ must bid for BS $j$ prior to BS $j'$. We conclude that user $k$ also has never bidden for BS $j'$. However, user $k$ is currently in the waiting list of BS $j'$, indicating that user $k$ has bidden for BS $j'$ before, which is a contradiction. Thus, we conclude that the repeated game converges.
\end{proof}

\begin{lemma}\label{lemma9}
During any round of the repeated game, if user $k$ bids for BS $j$, user $k$ cannot have a better choice than BS $j$.
\end{lemma}

\begin{proof}
According to the bidding strategy of users, if user $k$ bids for BS $j$, either BS $j$ is its most preferable choice or user $k$ has been rejected by another BS $j'$. The reason BS $j'$ rejects user $k$ is because it already held $S_j$ bids that are better than user $k$. Since the sequence of bids for each BS is non-decreasing, it is impossible for user $k$ to enter the waiting list of BS $j'$. Thus, user $k$ can not have a better choice than BS $j$.
\end{proof}

\begin{theorem}\label{theorem3}
The repeated game converges to a Nash equilibrium that is optimal for each user and BS.
\end{theorem}

\begin{proof}
It can be easily seen from Algorithm~\ref{alg:2} that each BS holds the set users with the maximum sum payments. For an SBS, if the sum of user payments is less than its power cost, the optimal strategy is to turn OFF so that the utility is increased from a negative value to zero.

From Lemma~\ref{lemma9}, if a user is currently in the waiting list of a BS, this BS is the best possible option for the user. Thus, when the game converges, the outcome is also optimal for each user.
Following Theorem~\ref{theorem2}, we conclude that the repeated bidding game converges to an NE.
\end{proof}

\section{Simulation Study \label{sec:sim}}

We evaluate the proposed centralized and distributed schemes with MATLAB simulations. We use the path loss and SINR models in~\cite{Bethanabhotla14}. A $1000$m $\times$ $1000$m area is considered. The massive MIMO BS is located at the center, the SBS's are randomly distributed in the area. We consider two cases for user distribution. In the first case, users are uniformly distributed across the area. In the second case, we divide the area into $8$ subareas, the number of users in each subarea is a Poisson random variable
and the traffic in each subarea is not uniformly distributed.
The power of the MBS and SBS's is set to $40$ dBm, we assume the available bandwidth for all the BS's is 1 MHz. The number of channels is $50$ for SBS's, thus $S_j=50$ for $j=1,\ldots,J$. We also set $S_0=100$.

We compare with two heuristic schemes for BS ON-OFF switching strategy. {\em Heuristic 1} is based on a load-aware strategic BS sleeping mode proposed in~\cite{Soh13}. Specifically, SBS $j$ is turned ON with probability $\min \left\{ \theta_j / S_j, 1 \right\}$, where $\theta_j$ is the number of users within the coverage of SBS $j$. {\em Heuristic 2} is based on a scheme presented in~\cite{Ashraf10}, where an SBS is activated
whenever there is a user enters its coverage area. We also consider the case that all the SBS's are always turned ON as a benchmark (termed Always On). For all the schemes, the user association strategy is based on the solution of problem ${\bf{P3}}$ for given BS ON-OFF states.

The EEs of different schemes are presented in Figs.~\ref{fig1}--\ref{fig4}. In Figs.~\ref{fig1} and~\ref{fig2}, it can be seen that the EEs of Always ON and Heuristic 2 schemes decreases when the number of SBS's becomes large, due to the fact that some SBS's are under-utilized while still consume power. The EEs of the proposed schemes and Heuristic 1 do not decrease as the number of SBS's grows, since these schemes can dynamically adjust to the traffic demand and turn OFF redundant SBS's. As expected, the centralized scheme achieves the highest energy efficiency. Note that the EE of Heuristic 2 is close to the Always ON scheme, since an SBS is easily activated when the numbers of users and SBS's are sufficiently large. Compare Figs.~\ref{fig1} and.~\ref{fig2}, it can be seen that when the traffic load is varying over subareas, the gaps between the proposed schemes and other schemes are slightly increased since larger gains can be achieved when the traffic demand becomes geographically dynamic. In Figs.~\ref{fig3} and~\ref{fig4}, we also find that the proposed schemes outperform the other schemes under different numbers of users, while the gaps become smaller as the number of users grows. This is because when the traffic load is increasing, activating more SBS's can effectively offload the traffic load from MBS and enhance the sum rate more significantly. Thus, the optimal BS ON-OFF strategy is close to that of Always ON, resulting in the reduced performance gap.

We also evaluate the sum rate of the schemes in Figs.~\ref{fig5}--\ref{fig8}. From Figs.~\ref{fig5} and~\ref{fig6}, we find that the sum rate is improved with more SBS's as a result of more offloading and high average SINR. Obviously, Always ON offers the best performance since it is possible for each user to connect to the BS with the largest achievable rate. The sum rate of the centralized scheme is close to that of Always ON. This is because we choose to turn OFF the SBS's that are not energy efficient, i.e., the sum rates of users connecting to these SBS's are not large enough and it is not worthy to turn ON these SBS's. The distributed scheme also achieves a high sum rate performance, because the SBS's with negative utility are turned OFF. Since the sum rates of these SBS's are relatively small, the performance loss is small. It is also observed that when the traffic is not uniformly distributed, the sum rates of the two heuristic schemes are slightly decreased since fewer SBS's are likely to be activated and more users are served by the MBS.

In Fig.~\ref{fig9}, an example of the repeated user bidding game is given. It can be seen that the game converges after a few number of rounds with the proposed algorithm. Note that, after a maximum utility is achieved after $6$ rounds, the utility of the BS's is increased again, due to the fact that some SBS's with negative utility are turned OFF. The utility of users is decreased since some SBS's are turned OFF and their users are handed over to the MBS.

\begin{figure*}[!t]
	\begin{minipage}[t]{.32\linewidth}
		\centering
		\includegraphics[width=2.25in, height=1.53in]{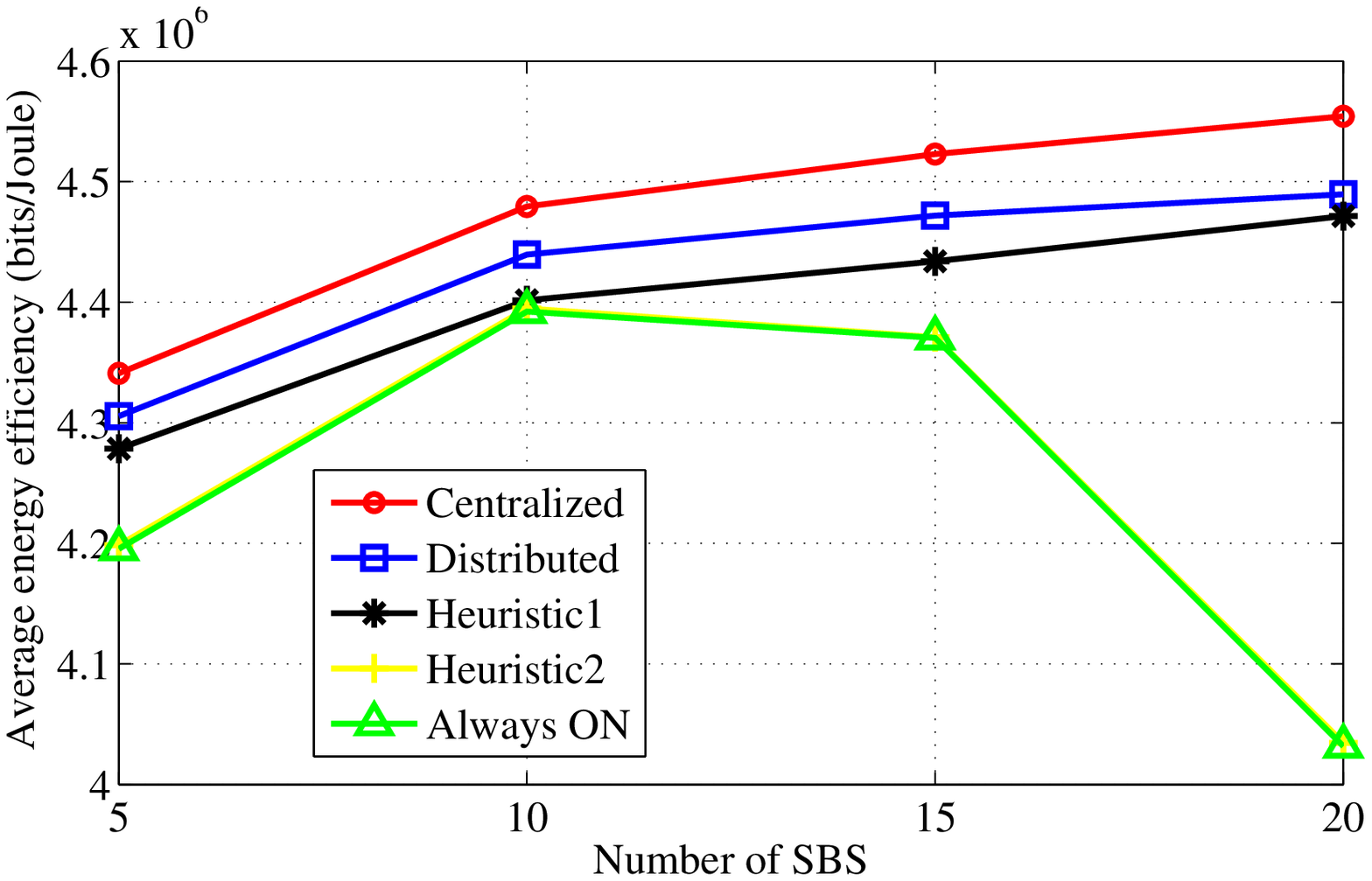}
		\caption{\footnotesize Average system EE versus number of SBS's for different
		BS ON-OFF switching strategies: $100$ users, uniformly distributed.}
		\label{fig1}
		\vspace{-0.1in}
	\end{minipage}
	\hspace{0.025in}
	\begin{minipage}[t]{.32\linewidth}
		\centering	
		\includegraphics[width=2.25in, height=1.53in]{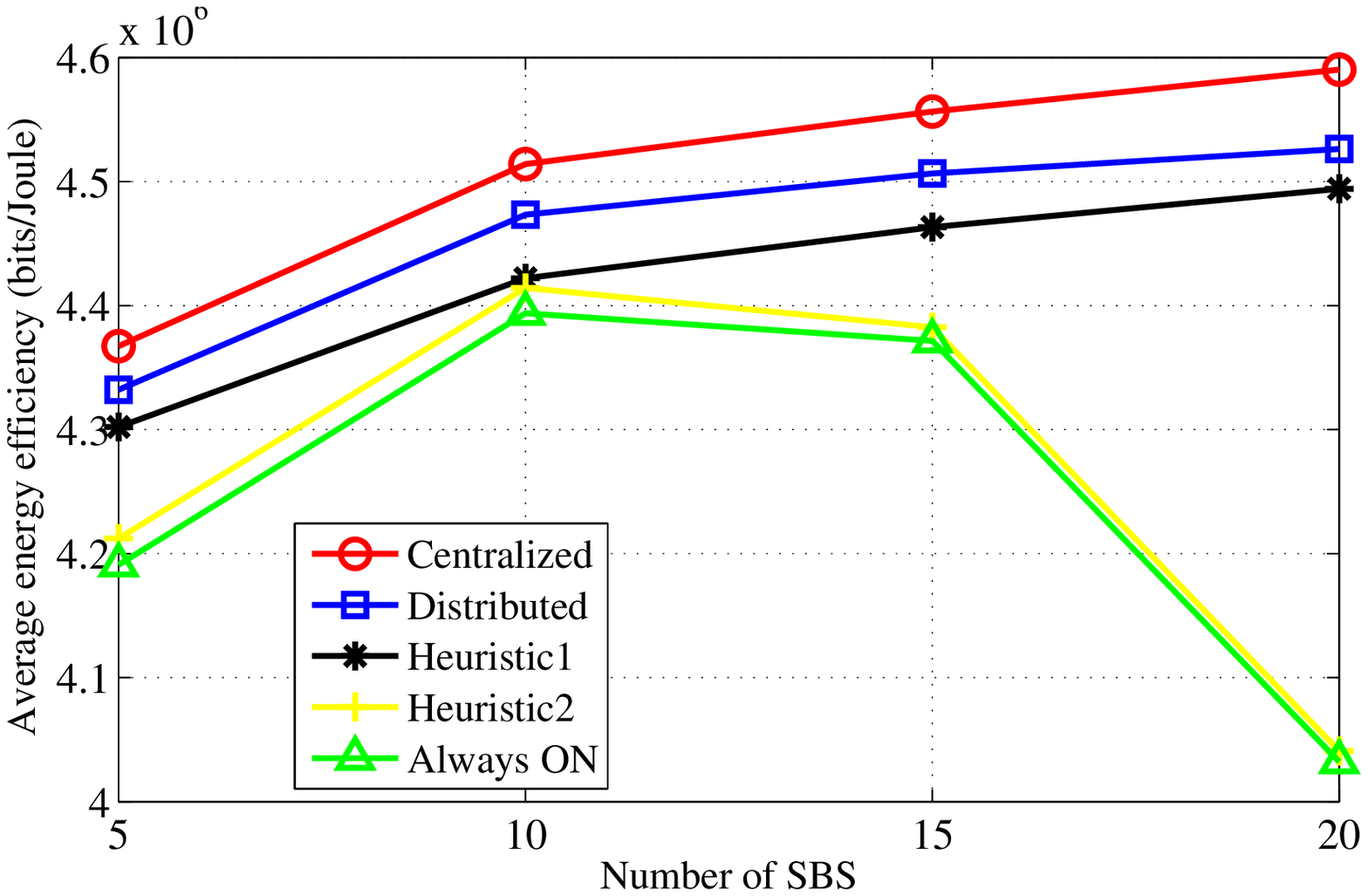}
		\caption{\footnotesize Average system EE versus number of SBS's for different BS
		ON-OFF switching strategies: $100$ users, non-uniformly distributed. }
		\vspace{-0.1in}
		\label{fig2}
	\end{minipage}		
	\begin{minipage}[t]{.32\linewidth}
		\centering	
		\includegraphics[width=2.25in, height=1.53in]{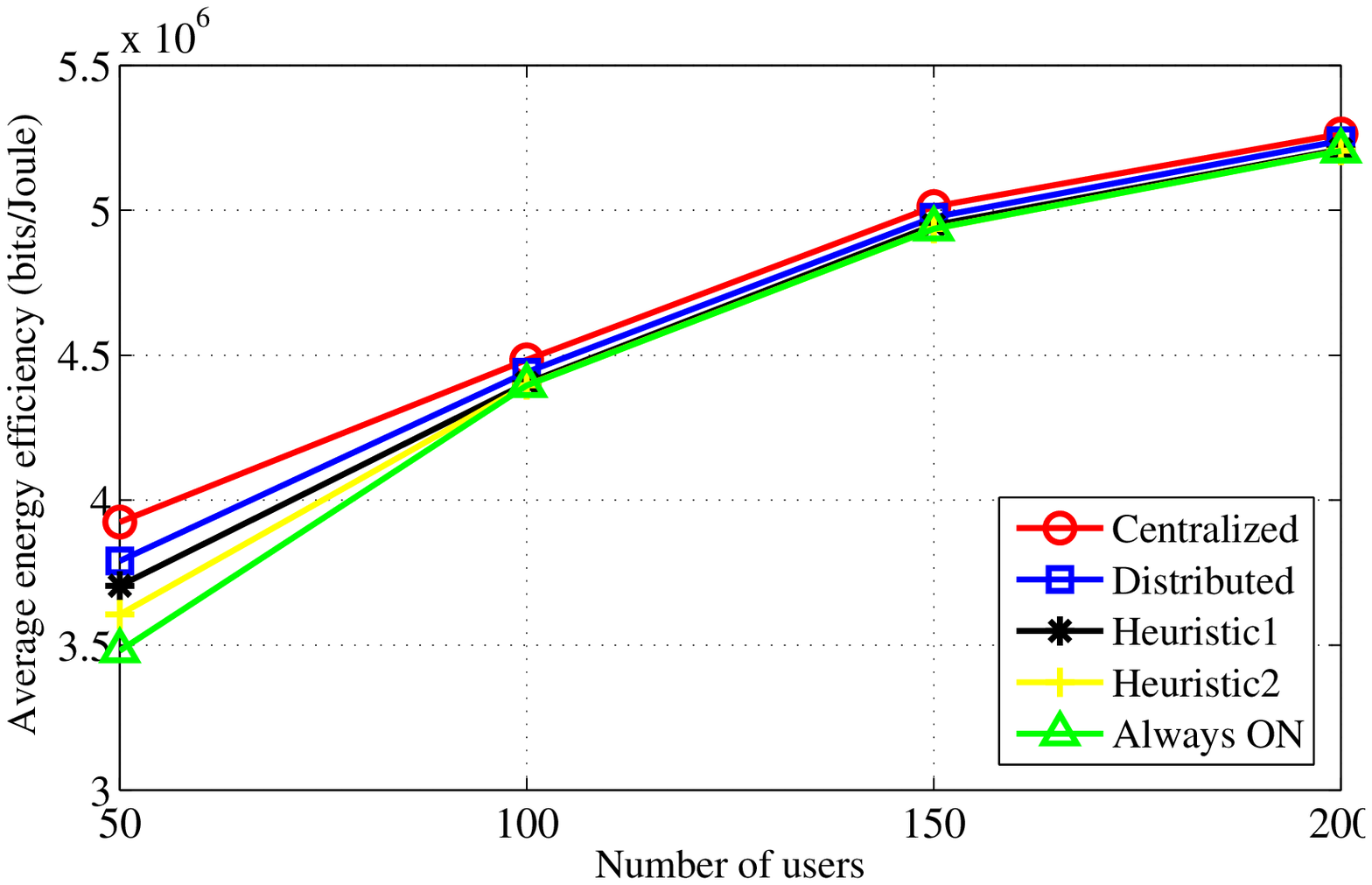}
		\caption{\footnotesize Average EE efficiency versus number of users for different
		BS ON-OFF switching strategies: uniformly distributed users, $10$ SBS's.}
		\vspace{-0.1in}
		\label{fig3}
	\end{minipage}	
\end{figure*}

\begin{figure*}[!t]
	\begin{minipage}[t]{.32\linewidth}
		\centering
		\includegraphics[width=2.25in, height=1.53in]{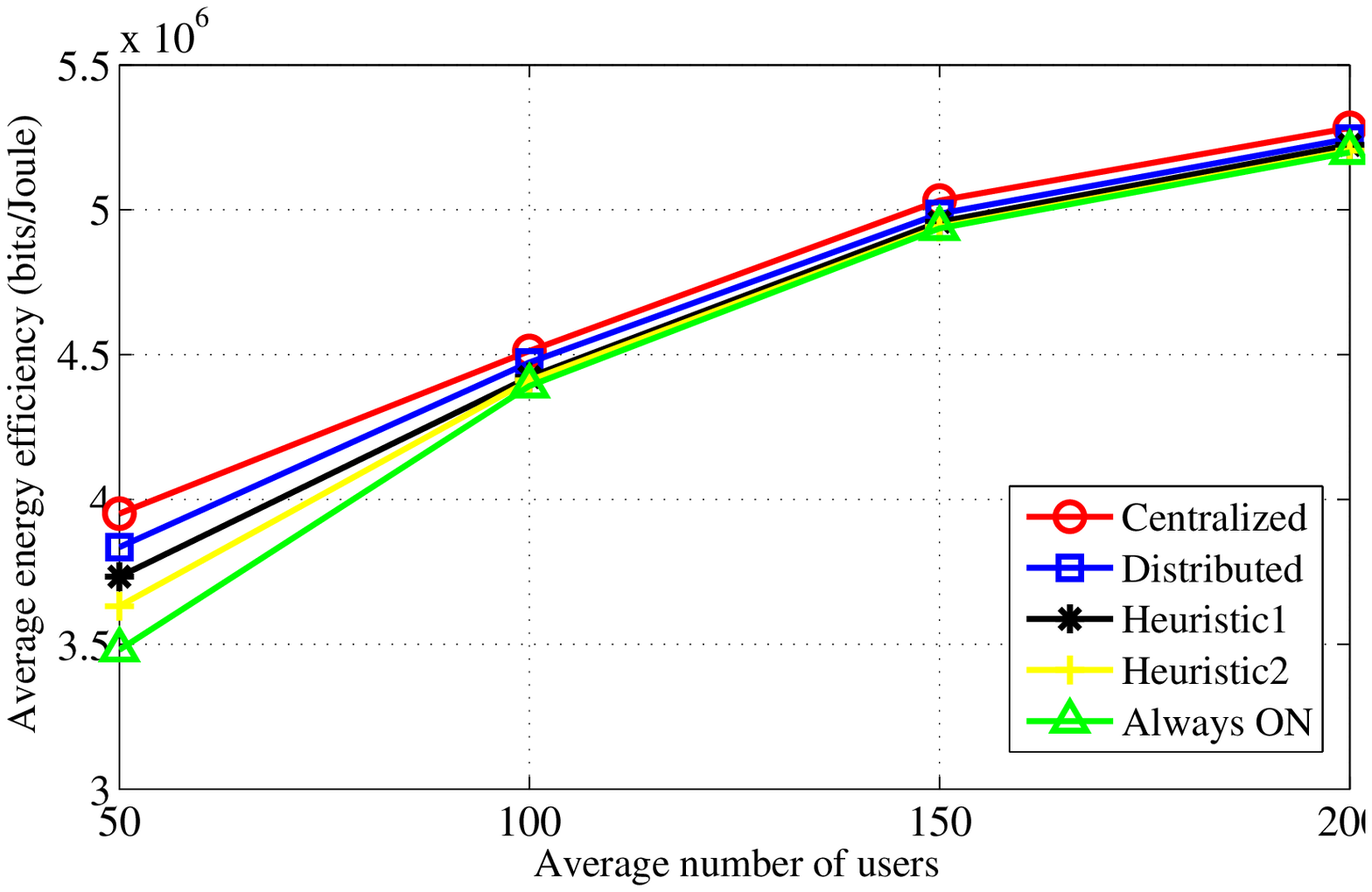}
		\caption{\footnotesize Average system EE versus average number of users for
		different BS ON-OFF switching strategies: non-uniformly distributed users, $10$ SBS's.}
		\label{fig4}
		\vspace{-0.1in}
	\end{minipage}
	\hspace{0.025in}
	\begin{minipage}[t]{.32\linewidth}
		\centering	
		\includegraphics[width=2.25in, height=1.53in]{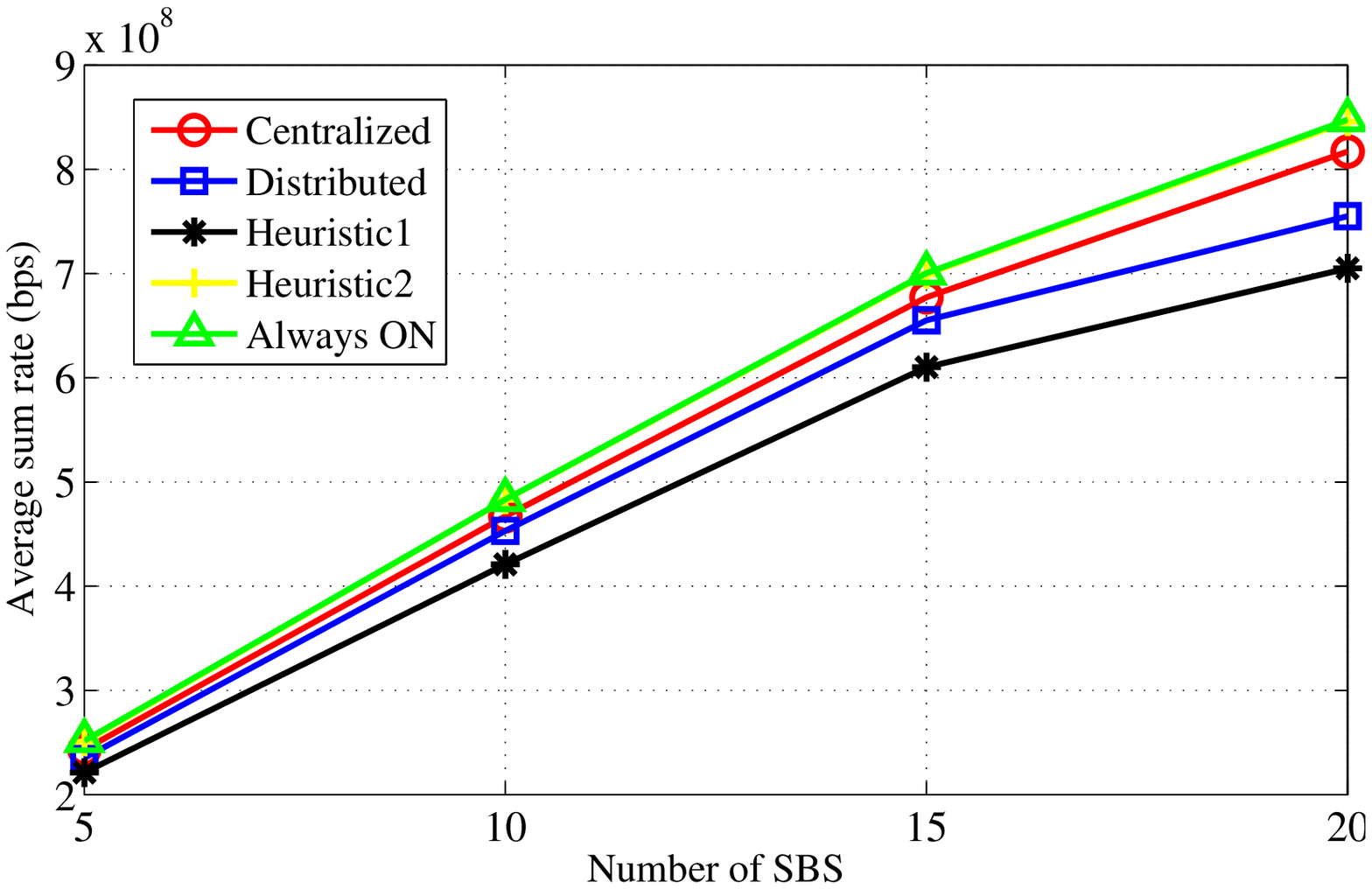}
		\caption{\footnotesize Average sum rate versus number of SBS's for different BS OF-OFF switching
		strategies: $100$ users, uniformly distributed. }
		\vspace{-0.1in}
		\label{fig5}
	\end{minipage}		
	\begin{minipage}[t]{.32\linewidth}
		\centering	
		\includegraphics[width=2.25in, height=1.53in]{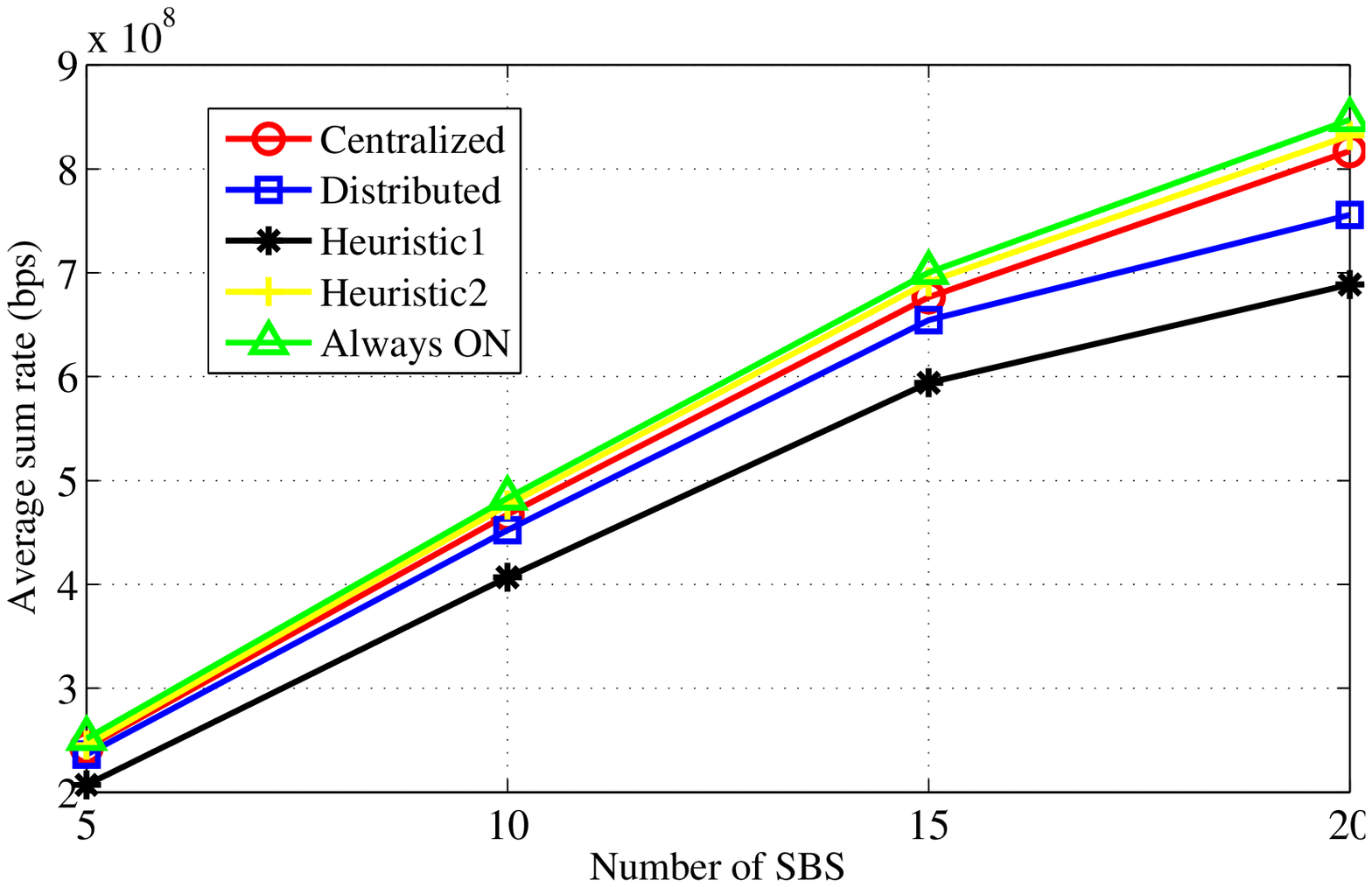}
		\caption{\footnotesize Average sum rate versus number of SBS's for different BS ON-OFF switching
		strategies: $100$ users, non-uniformly distributed.}
		\vspace{-0.1in}
		\label{fig6}
	\end{minipage}	
\end{figure*}

\begin{figure*}[!t]
	\begin{minipage}[t]{.32\linewidth}
		\centering
		\includegraphics[width=2.25in, height=1.53in]{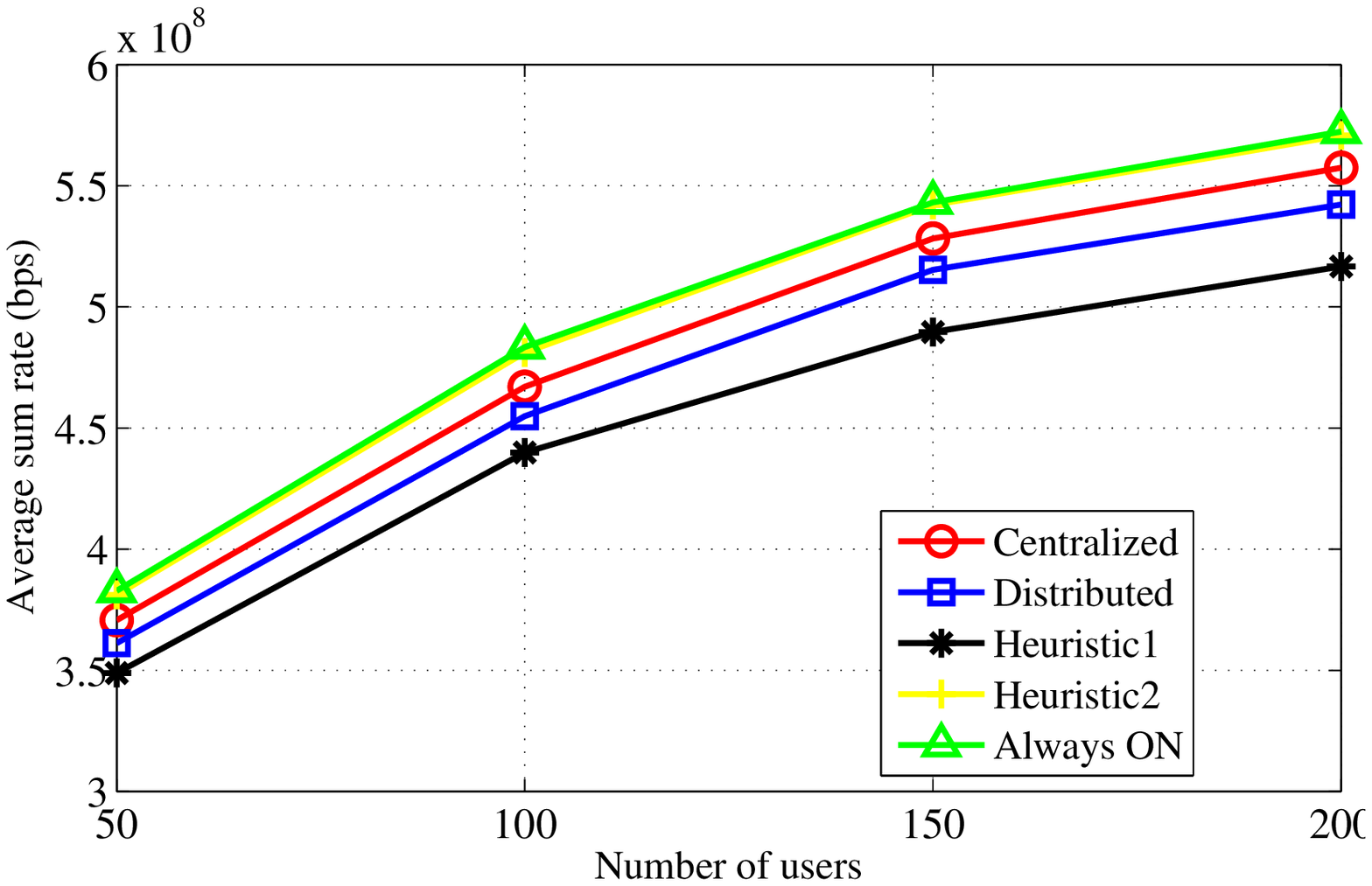}
		\caption{\footnotesize Average sum rate versus number of users for different BS ON-OFF switching
		strategies: uniformly distributed users, $10$ SBS's.}
		\label{fig7}
		\vspace{-0.1in}
	\end{minipage}
	\hspace{0.025in}
	\begin{minipage}[t]{.32\linewidth}
		\centering	
		\includegraphics[width=2.25in, height=1.53in]{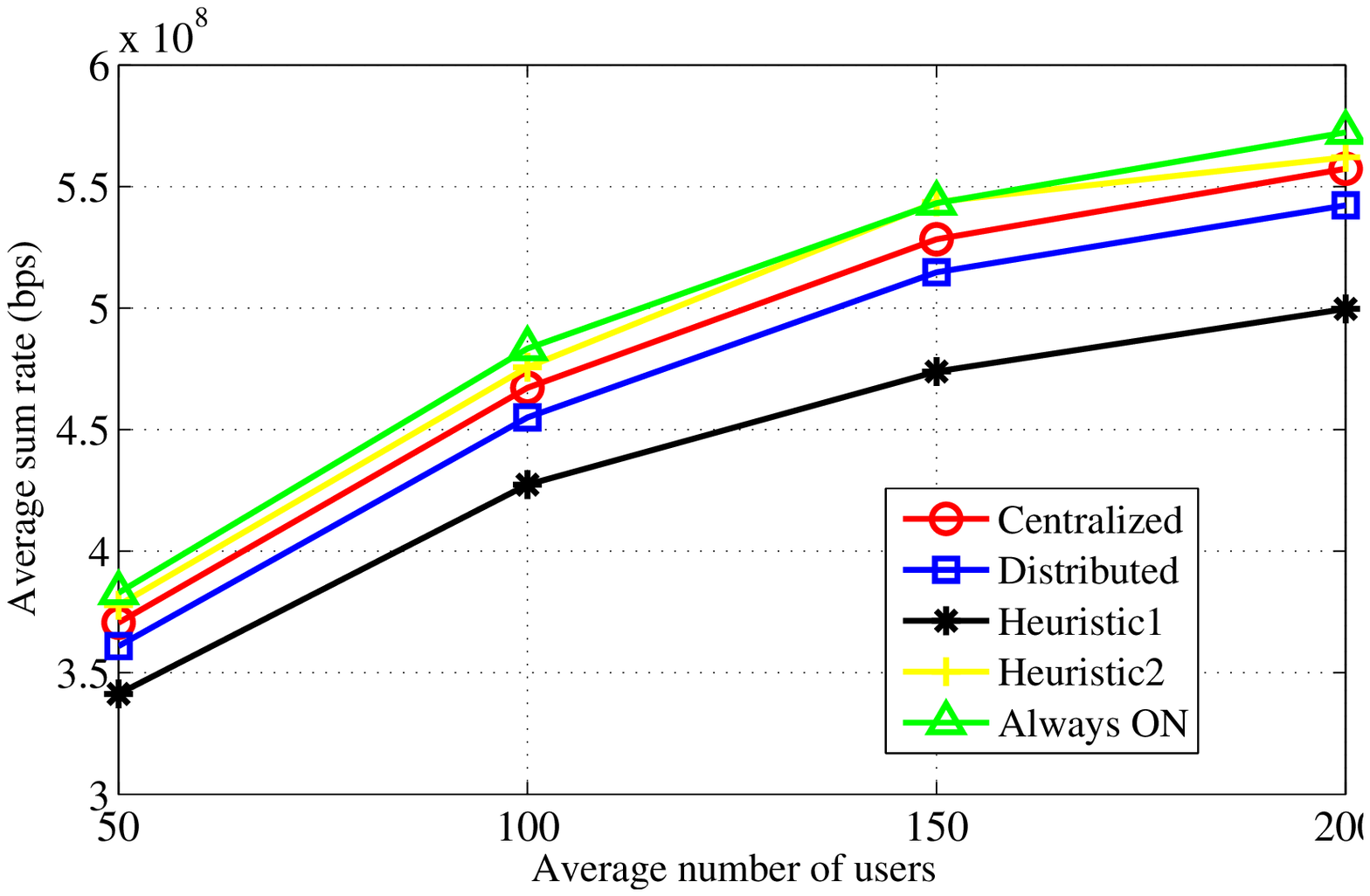}
		\caption{\footnotesize Average sum rate versus number of users for different BS ON-OFF switching
		strategies: non-uniformly distributed users, $10$ SBS's.}
		\vspace{-0.1in}
		\label{fig8}
	\end{minipage}		
	\begin{minipage}[t]{.32\linewidth}
		\centering	
		\includegraphics[width=2.25in, height=1.5in]{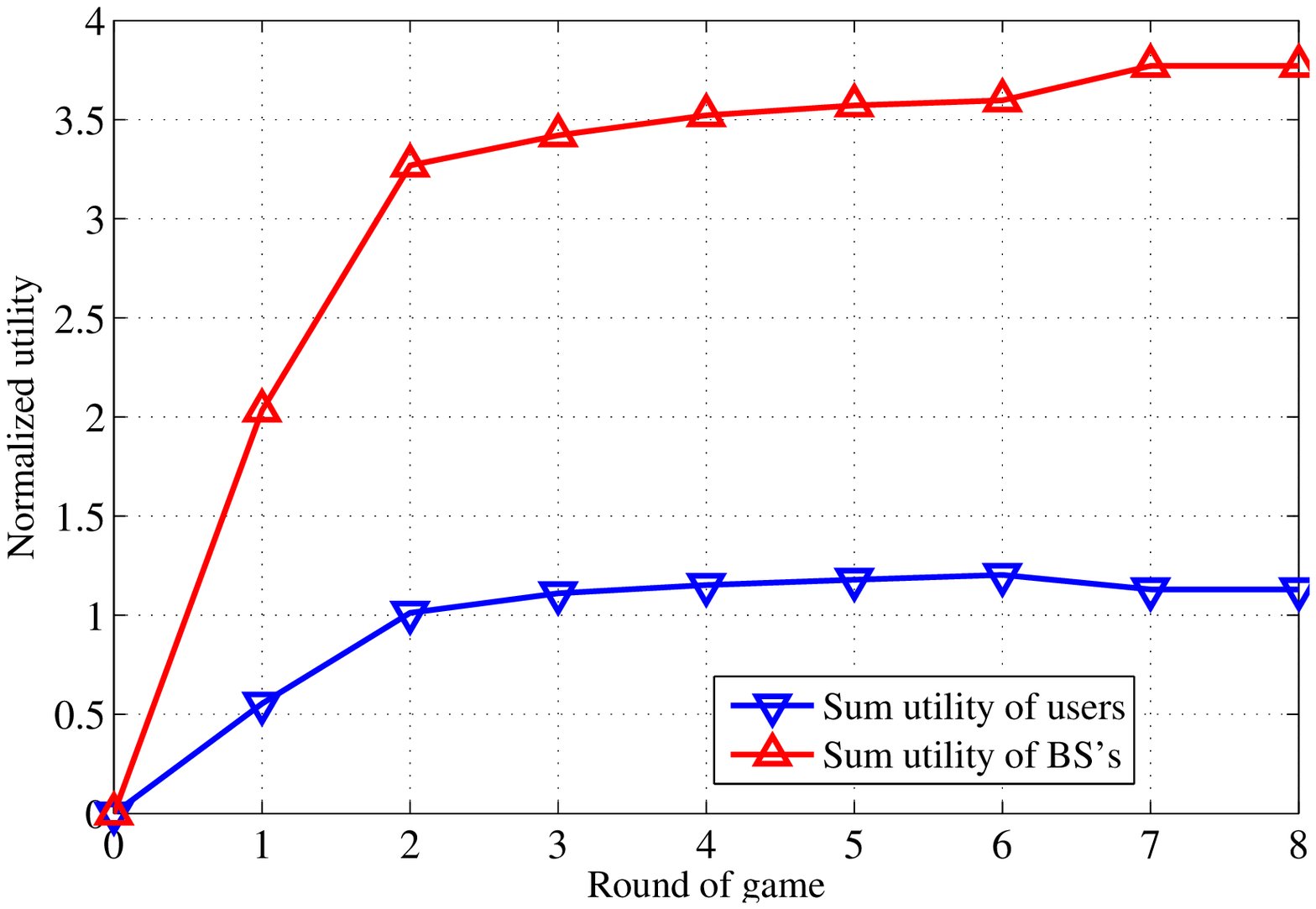}
		\caption{\footnotesize Convergence of the repeated bidding game: $100$ users and $10$ SBS's.}
		\vspace{-0.1in}
		\label{fig9}
	\end{minipage}	
\end{figure*}

\section{Related Work \label{sec:related}}

As key technologies for 5G network, massive MIMO and small cells have been extensively studied in prior works. The fundamental PHY layer techniques of massive MIMO were introduced in~\cite{Marzetta10,Rusek13}. Beyond the PHY, upper layer technologies in a wireless network with massive MIMO are also considered,
such as~\cite{Feng16,Xu15,Fernandes13}. In~\cite{Xu15}, user association and resource allocation in a massive MIMO HetNet were investigated with the objectives of rate maximization and rate maximization with proportional fairness.
In~\cite{Fernandes13}, a time-shift frame structure was proposed to mitigate inter-cell interference caused by pilot contamination in a multi-cell massive MIMO system. Since neighboring cells transmit pilots at different time instants, the inter-cell interference can be well
mitigated.

User association in HetNet has been widely investigated, such as~\cite{Ye13,Zhou15,Feng12}. In~\cite{Ye13}, user association and resource allocation were jointly considered to maximize the sum utility of users. Using dual decomposition, the proposed scheme can be implemented with a distributed algorithm, and the solution is shown to be near-optimal. In~\cite{Zhou15}, user association was considered to minimize the maximum load among all BS's, several approximation algorithms were proposed with analysis on complexity and performance bound. In~\cite{Feng12}, user association is determined by the achievable rate of each user. The HetNet with dense small cell deployment has drawn increasing interests since new challenges arise when a large number of SBS's are deployed in a given area. An overview of hyper-dense HetNet was presented in~\cite{Hwang13}, several cooperative approaches were proposed in~\cite{Xu14} and~\cite{Feng14} to enhance the network performance.

The EE has become an important objective for wireless networks in recent years. Specifically, the designs of energy-efficient massive MIMO systems were studied in~\cite{Bjornson13, Liu15, Ng12, Li14, Nguyen15}, where power and subcarrier allocation, antenna selection, and pilot allocation were considered to maximize EE. Compared to these works, we focus on the optimal BS ON-OFF switching strategy for energy efficient massive MIMO enabled HetNets. Some prior works also aimed to improve the EE of HetNets~\cite{Soh13, Ashraf10}. In~\cite{Soh13}, the authors considered two sleeping strategies for MBS, and derived the success probability and EE for a $K$-tier heterogeneous network using stochastic geometry analysis. In~\cite{Ashraf10}, a BS sleeping strategy was proposed to improve the EE of femtocell network. The BS ON-OFF switching strategie was also investigated by prior works including~\cite{Oh13}. In~\cite{Oh13}, a distributed algorithm that is easy to implement was proposed. The key principle is the use of a new notion called {\em network-impact}, which accounts for the load increments brought to other BS's by turning a BS OFF.

\section{Conclusions \label{sec:con}}

In this paper, we considered optimal BS ON-OFF switching and user association to maximize the EE
of a
massive MIMO HetNet. We formulated an
integer programming problem and proposed a centralized scheme to solve it with proven optimality. We also proposed a distributed scheme based on the user bidding approach and showed that an NE
can be achieved for each user and BS. The proposed schemes were evaluated with simulations and the results demonstrated their superior performance over several benchmark schemes.

-------------------------------------------------------------
\section*{Acknowledgment}
-------------------------------------------------------------
%
%


\end{document}